\documentclass[a4paper]{article}

\usepackage{amsmath,amsfonts,amssymb,amsthm}

\newtheorem{defi}{Definition}
\newtheorem{teo}{Theorem}

\newtheorem{prop}[teo]{Proposition}
\newtheorem{cor}[teo]{Corollary}
\newtheorem{lem}[teo]{Lemma}
\newtheorem{rmk}{Remark }
\newtheorem{exm}{Example}

\begin{document}
\title{Extensions of Hamiltonian systems dependent on a rational parameter}
\author{Claudia M. Chanu,  Luca Degiovanni, Giovanni Rastelli \\ \\
Dipartimento di Matematica G.\ Peano,\\ Universit\`a di Torino.   Torino, via Carlo Alberto 10, Italia.\\ \\ e-mail: claudiamaria.chanu@unito.it \\ luca.degiovanni@gmail.com \\ giorast.giorast@alice.it }
\maketitle

\begin{abstract}
The technique of ``extension" allows to build $(n+1)$-dimensional Hamiltonian systems with a non-trivial polynomial in the momenta first integral of any given degree starting from a 
$n$-dimensional Hamiltonian satisfying some additional properties.
Until now, the application of the method was restricted to integer values of a certain fundamental parameter determining the degree of the additional first integral. In this article we show how this technique can be generalized to any rational value of the same parameter. Several examples are given, among them the anisotropic oscillator and a special case of the Tremblay-Turbiner-Winternitz system.
\end{abstract}  

\section{Introduction}

The ``extensions" of  natural Hamiltonians  were introduced in \cite{CDRfi} as a tool for building Hamiltonian systems with polynomial first integrals of any given degree. The technique of extension was developed  as the generalization of an iterative procedure to generate a first integral of degree $\lambda\in \mathbb{N}-\{0\}$ in the momenta \cite{Pol} for
the system
\begin{equation} \label{cal0}
\frac 12 p_r^2+\frac 1{r^2}\left(\frac 12p_\psi^2+\frac a{\sin^2(\lambda\psi)}\right),
\end{equation}
 This superintegrable system was already considered in \cite{CDR0}, where it was remarked the existence of additional polynomial first integrals for several values of $\lambda\in \mathbb Q$ and a general form for odd positive integers was conjectured. Then, a generalization of the system was introduced by Tremblay, Turbiner and Winternitz (TTW) \cite{TTW} and, subsequently, studied by many authors (see for instance \cite{KKM,MPY,MPW} and references therein), proving the superintegrability of the TTW system, and of several related new systems, for any rational value of $\lambda$.

The technique of extensions was generalized in \cite{CDR}
and employed to build new superintegrable systems from existing ones in \cite{CDR1}.
The extension theory can be easily generalized to complex manifolds, see [7]
for details.

In this paper we  develop a more general iterative procedure that in particular allows us to construct, for any positive rational $\lambda=m/n$, a
polynomial first integral of degree $m+n-1$  of the system (\ref{cal0}), as well as of the two uncoupled harmonic oscillators with rational ratio $\lambda$ of the frequencies
\begin{equation}\label{osc0}
\frac{1}{2}p_y^2+\frac{1}{2}p_x^2+\omega^2x^2+\frac{\omega^2}{\lambda^2}y^2.
\end{equation}
For both systems, the case of integer $\lambda$ is recalled in Sec.\ 2, where
the extension technique is briefly summarized.  Its generalization is exposed in Sec.\ 3, while
Sec.\ 4 contains several detailed examples. In the Appendix, the trigonometric tagged functions, employed extensively in the article, are defined and some of their properties are shown.

\section{Extension of Hamiltonian systems}

The extension procedure has been characterized in \cite{CDRfi,CDR}.
Given a Hamiltonian $L$ on a $2n$-dimensional Poisson manifold $Q$, we construct the
$(2n+2)$-dimensional Poisson manifold $M=T\times Q$, where $T$ is the cotangent bundle of a one-dimensional manifold with the canonical symplectic form $\mathrm{d}p_u\wedge \mathrm{d}u$.

The main result (Proposition 1 in \cite{CDR}) states that a Hamiltonian $L$ on the Poisson manifold $Q$, with Hamiltonian vector field $X_L$, admits an extension of the form
$$
\frac{1}{2}p_u^2+\alpha(u)L+\beta(u),
$$
with a polynomial first integral of the form $U^m(G),$ where
$$
U=p_u+\gamma(u)X_L,
$$
if and only if there exists a function $G$ on $Q$  that satisfies the relation
\begin{equation}\label{CG_old}
X_L^2(G)=-2m(cL+L_0)G,
\end{equation}
for some $m\in\mathbb{N}-\{0\}$ and $c,L_0\in\mathbb{R}$. If a solution is found for $(mc,L_0)$ with $c\neq 0$, then, without loss of generality we can set $L_0$=0.
When the relation (\ref{CG_old}) holds,  the function $\gamma$ is any solution of
\begin{equation}\label{g_old}
\gamma''+2c\gamma\gamma'=0
\end{equation}
and the functions $\alpha$ and $\beta$ are calculated directly from $\gamma$ through the relations \begin{equation}\label{albe}
\alpha=-m\gamma',\quad \beta=mL_0\gamma^2. 
\end{equation}
The form of $\gamma$, and therefore of the operator $U$, depends on the value of $c$. If $c\neq0$, then for any $\kappa\in\mathbb{R}$ a solution of (\ref{g_old}) and the corresponding operator are 
$$
\gamma=\frac{1}{T_\kappa(cu)},\qquad U=p_u+\frac{1}{T_\kappa(cu)}X_L.
$$
The choice $L_0=0$ leads to $\beta=0$,  hence, up to inessential constants, the extended Hamiltonian is
$$
\frac{1}{2}p_u^2+\frac{mc}{S^2_\kappa(cu)}L  
$$
where the trigonometric tagged functions $S_\kappa$, $T_\kappa$ are defined in the Appendix.
Conversely, if $c=0$, then for any $A\in\mathbb{R}$
$$
\gamma=-Au, \qquad U=p_u-AuX_L
$$
and the extended Hamiltonian, written  up to inessential constants, becomes
$$
\frac{1}{2}p_u^2+mAL+mL_0A^2u^2.
$$
The parameters in equation (\ref{CG_old}) and (\ref{g_old}) are chosen in order to obtain a function $\gamma$, and therefore an operator $U$, not depending on $m$. This choice, going back to \cite{CDRfi}, gives a simpler form for $\gamma$ but has a drawback: $G$ seems to depend on the three independent parameters $m$, $c$ and $L_0$ instead that on the two parameters $mc$ and $mL_0$, as in fact it is. 
By setting instead $\tilde{c}=mc$ and $\tilde{L}_0=mL_0$,  equations (\ref{CG_old}) and (\ref{g_old}) become respectively
\begin{eqnarray}
&&X_L^2(G)=-2(\tilde{c}L+\tilde{L}_0)G, \label{CG}\\
&&m\gamma''+2\tilde{c}\gamma\gamma'=0. \label{g_new}
\end{eqnarray}
Within this notation, $G$ explicitly depends  on  two parameters only and it is clear that $G$ remains the same for different values of $m$, but $\gamma$ and $U$ become dependent on $m$. A simple calculation shows that if $\gamma(u;c)$ satisfies (\ref{g_old}) then $m\gamma(u;\tilde{c})$ satisfies (\ref{g_new}); moreover,  relations (\ref{albe}) hold unchanged.

In the following, the choice of parameters made in equations (\ref{CG}) and (\ref{g_new}) will be used and therefore, for $\tilde{c}\neq0$, $\kappa\in\mathbb{R}$,
\begin{equation}\label{ga1}
\gamma=\frac{m}{T_\kappa(\tilde{c}u)},\qquad U=p_u+\frac{m}{T_\kappa(\tilde{c}u)}X_L,
\end{equation} 
and the extended Hamiltonian becomes
\begin{equation}\label{par1}
\frac{1}{2}p_u^2+\frac{m^2\tilde{c}}{S^2_\kappa(\tilde{c}u)}L . 
\end{equation}
Conversely, for $\tilde{c}=0$, $A\in\mathbb{R}$, we have
\begin{equation}\label{ga2}
\gamma=-mAu, \qquad U=p_u-mAuX_L
\end{equation}
and the extended Hamiltonian becomes
\begin{equation}\label{par2}
\frac{1}{2}p_u^2+m^2AL+m^2\tilde{L}_0A^2u^2.
\end{equation}
When $L$ is a natural Hamiltonian, and therefore $Q$ is the cotangent bundle of a Riemannian manifold, the configuration manifolds of the extended Hamiltonians are warped manifolds. 

Several explicit examples of extensions are exposed into details in \cite{Pol,CDRfi,CDR1}.
\begin{exm} \rm (See also \cite{Pol}) 
We apply the extension procedure  to the Hamiltonian
\begin{equation}\label{calL}
L=\frac 12p_\phi^2+\frac {a}{\sin^2(\phi)}.
\end{equation}
One of the possible extensions (with $\kappa=0$) is 
\begin{equation}\label{ex1}
H^{\mathrm{cal}}=\frac{1}{2}p_u^2+\frac{\lambda^2}{u^2}\left(\frac{1}{2}p_\phi^2+\frac{a}{\sin^2\phi}\right),
\end{equation}
which coincides with the Calogero-type system (\ref{cal0}) after the rescaling $\phi=\lambda\psi$.
In this example $\tilde{c}\neq0$ and $\lambda^2=m^2/\tilde{c}$. The relation (\ref{CG}) imposes a very strong constraint on the parameters: if the function $G$ does not depend on the momenta, then (\ref{CG}) has a solution only for $\tilde{c}=1$ and $\tilde{L}_0=0$. Hence, $\lambda^2=m^2$ and $\lambda$ is necessarily an integer.
\end{exm}
\begin{exm}\rm (See also \cite{CDR1})
We consider the system of two uncoupled harmonic oscillators (\ref{osc0}). 
By rescaling $u=\lambda y$ and dividing the Hamiltonian by the constant factor $\lambda^2$, we get the equivalent 
Hamiltonian
\begin{equation}\label{ex2}
H^{\mathrm{osc}}=\frac{1}{2}p_u^2+\lambda^2\left(\frac{1}{2}p_x^2+\omega^2x^2\right)+\lambda^2\omega^2u^2.
\end{equation}
Here we have an extension of the one-dimensional harmonic oscillator with $\tilde{c}=0$, $\lambda^2=m^2A$ and $\omega^2=A\tilde{L}_0$. Even in this case, the equation (\ref{CG}) admits a solution not depending on the momenta only for $\tilde{L}_0=\omega^2$. This constraint implies $A=1$ and $\lambda^2=m^2$, forcing $\lambda$ to be an integer.
\end{exm}

Both  systems of the above examples admit the maximal number of functionally independent first integrals, namely, $H$, $L$ and $U^\lambda G$, for $\lambda$ integer, so that they are superintegrable. However, it
 is well known that Hamiltonians (\ref{ex1}) and (\ref{ex2}) are superintegrable for any rational value of the parameter $\lambda$ (\cite{MPY,MPW,KKM}).  In the following we show an effective procedure to construct a third independent first integral for non integers values of $\lambda$ by means of a more general extension procedure.

\section{Extensions with rational parameters}
It is useful to slightly modify  the  expression of the extensions introduced so far and rename them as follows

\begin{defi}
The Hamiltonian
\begin{equation}\label{m_ext}
H_m=\frac{1}{2}p_u^2+m^2\tilde \alpha(u) L+m^2\tilde \beta(u), 
\end{equation}
defined on the Poisson manifold $T \times Q$, is a $m$-\textbf{extension} of $L$ generated by $G$, if
\begin{equation}\label{m_int}
K_m=(p_u+m\tilde \gamma(u)  X_L)^m(G)
\end{equation}
is a first integral of $H_m$.
\end{defi}
The comparison between (\ref{ga1})-(\ref{par2}) and (\ref{m_ext})-(\ref{m_int}) shows that $\tilde \alpha$, $\tilde \beta$, $\tilde \gamma$ are independent of $m$ and that in particular the following relations hold
$$
\tilde \alpha=\frac \alpha{m^2}, \quad \tilde \beta=\frac \beta{m^2}, \quad \tilde \gamma=\frac \gamma m,
$$ 
(see Table 1 for their expressions according to the value of $\tilde c$).
Moreover, $G$ must satisfy (\ref{CG}). 
\begin{table}[h]
\begin{center}
\begin{tabular}{|c|c|c|}
\hline
&$\tilde c=0$ & $\tilde c \neq 0$ \cr

\hline

$\tilde \alpha = -\tilde{\gamma}'= $&$A$ &  $\dfrac {\tilde c \vphantom{\frac 12} }{S_\kappa^2(\tilde c u)}$ \cr

$\tilde \beta = \tilde L_0 \tilde{\gamma}^2=$ &$\vphantom{\dfrac 12}\tilde L_0A^2u^2$ & $0$  \cr

$\tilde \gamma = $& $-Au$ & $\dfrac 1{T_\kappa(\tilde c u)}$\cr  

\hline
\end{tabular}
\end{center}
\caption{Functions involved in the $m$-extension of $L$}
\end{table}

In order to write as extensions the Hamiltonians (\ref{ex1}), (\ref{ex2}) with $\lambda=\frac mn$, one, naively, should divide $\tilde \alpha$, $\tilde \beta$, $\tilde \gamma$ by $n^2$. This corresponds, roughly speaking, to  a $m$-extension of $L/n^2$. More precisely, 

\begin{defi}\label{def2}
The Hamiltonian
\begin{equation}\label{mn_ext}
H_{m,n} = \frac{1}{2}p_u^2+\frac{m^2}{n^2}\tilde \alpha(u) L+\frac{m^2}{n^2}\tilde \beta(u), 
\end{equation}
defined on the Poisson manifold $T \times Q$,
is called a $(m,n)$-\textbf{extension} of $L$ generated by a function $G_n$ if
\begin{equation}\label{mn_int}
K_{m,n}=\left(p_u+\frac{m}{n^2}\tilde \gamma(u)   X_L\right)^m(G_n)
\end{equation}
is a first integral of $H_{m,n}$, where $\tilde \alpha$, $\tilde \beta$, $\tilde \gamma$ are defined as in Table 1.
\end{defi}

\begin{lem}\label{Uno}
A Hamiltonian $L$ admits a $(m,n)$-extension generated by a function $G_n$, if and only if the function $G_n$ satisfies
\begin{equation}\label{CGn}
X_L^2(G_n)=-2n^2(\tilde{c}L+\tilde{L}_0)G_n.
\end{equation}
\end{lem}
\begin{proof}
The Hamiltonian $L^{(n)}=\displaystyle{\frac{1}{n^2}L}$ admits a $m$-extension if and only if there exists a function $G_n$ satisfying
$$
X_{L^{(n)}}^2(G_n) = -2\left(\tilde{c}L^{(n)}+\tilde{L}_0^{(n)}\right)G_n.
$$
The $m$-extension of $L^{(n)}$ is then given by
\begin{equation}\label{ext_Ln}
\frac{1}{2}p_u^2+m^2\tilde \alpha L^{(n)}+m^2\tilde \beta ^{(n)},
\end{equation}
with first integral
$$
\left(p_u+m\tilde \gamma X_{L^{(n)}}\right)^m(G_n),
$$
where $\tilde \beta ^{(n)}=\tilde{L}_0^{(n)}\tilde \gamma  ^2$ and $\tilde \alpha$, $\tilde \beta$, $\tilde \gamma$  are given by Table 1. The definition of $L^{(n)}$ implies
$$
\frac{1}{n^4}X_L^2(G_n)=-2\left(\frac{\tilde{c}}{n^2}L+\tilde{L}_0^{(n)}\right)G_n.
$$
Hence, after setting $\tilde \beta =\tilde{L}_0\tilde \gamma  ^2$, one has
$$
\tilde \beta ^{(n)}=\tilde{L}_0^{(n)}\tilde \gamma  ^2=\frac{1}{n^2}\tilde{L}_0\tilde \gamma  ^2=\frac{1}{n^2}\tilde \beta 
$$
and the extension (\ref{ext_Ln}) becomes
$$
\frac{1}{2}p_u^2+\frac{m^2}{n^2}\tilde \alpha L+\frac{m^2}{n^2}\tilde \beta 
$$
which coincides with the $(m,n)$-extension of $L$.
\end{proof}

\begin{rmk}\rm  The Proof of Lemma \ref{Uno} shows that it is not restrictive to assume in Definition \ref{def2} that $\tilde \alpha$, $\tilde \beta$, $\tilde \gamma$ are those given in Table 1.
\end{rmk}
Consequently, the search for extensions with rational parameter reduces to the search for solutions of  (\ref{CGn}). Theorem \ref{Due} shows how to construct iteratively solutions $G_n$ of
(\ref{CGn}) starting from a known solution $G$ of (\ref{CG}).
\begin{teo}\label{Due}
Let $G$ be a function on $Q$ satisfying
$$
X_L^2(G)=\Lambda G
$$
with $X_L(\Lambda)=0$, then the recursion
\begin{equation}\label{rec}
G_1=G, \qquad G_{n+1}=X_L(G)\,G_n+\frac{1}{n}G\,X_L(G_n), 
\end{equation}
satisfies, for any $n\in \mathbb N-\{0\},$
$$
X_L^2(G_n)=n^2\Lambda G_n.
$$
\end{teo}
\begin{proof}
By induction on $n$.
For $n=1,$ one has $G_2=2G\,X_L(G)$ and the relation $X_L^2(G_2)=4\Lambda G_2$ is straightforward. Then,  let us assume that $G_n$ satisfies $X_L^2(G_n)=n^2\Lambda G_n$. Recalling that for two functions $A$ and $B$ the formula
$$
X_L^2(AB)=X_L^2(A)B+2X_L(A)X_L(B)+AX_L^2(B)
$$
holds, we get 
\begin{eqnarray*}
X_L^2(G_{n+1}) &=& X_L^2\left[X_L(G)\,G_n+\frac{1}{n}G\,X_L(G_n)\right]\\
&=& X_L^3(G)G_n+2X_L^2(G)X_L(G_n)+X_L(G)X_L^2(G_n)+\\
&& \frac{1}{n}\Big[X_L^2(G)X_L(G_n)+2X_L(G)X_L^2(G_n)+GX_L^3(G_n)\Big]\\
&=& \Lambda X_L(G)G_n+2\Lambda GX_L(G_n)+n^2\Lambda X_L(G)G_n+\\
&&\frac{1}{n}\Big[\Lambda GX_L(G_n)+2n^2\Lambda X_L(G)G_n+n^2\Lambda GX_L(G_n)\Big]\\
&=&(1+2n+n^2)\Lambda X_L(G)G_n+\frac{1+2n+n^2}{n}\Lambda GX_L(G_n)\\
&=&(n+1)^2\Lambda G_{n+1}.
\end{eqnarray*}
\end{proof}

\begin{cor}
A Hamiltonian $L$ admits a $m$-extension $H_m$ if and only if it admits a $(m,n)$-extension $H_{m,n}$ for any positive rational $m/n$.
\end{cor}
\begin{proof}
If $L$ admits $m$-extensions, then there exists $G$  satisfying (\ref{CG}).
By Theorem \ref{Due} with
\begin{equation}\label{Lamb}
\Lambda=-2(\tilde{c}L+\tilde{L}_0),
\end{equation}
 we can construct $G_n$ verifying the condition (\ref{CGn}) for any $n$.  Hence, we can construct the $(m,n)$-extension (\ref{mn_ext}). The converse is straightforward. 
\end{proof}

\begin{prop}
The closed form for $G_n$ satisfying the recursion (\ref{rec}) is
\begin{equation}\label{Gn_alt}
G_n=\sum_{k=0}^{\left[\frac{n-1}{2}\right]}
\binom{n}{2k+1}\Lambda^k G^{2k+1}(X_L G)^{n-2k-1},
\end{equation} 
where $[\cdot]$ denotes the integer part.
\end{prop}

\begin{proof}
Since
\begin{eqnarray*}
X_L(G^{2k+1}(X_L G)^{n-2k-1})&=&(2k+1) G^{2k}(X_L G)^{n-2k} \\
&&+(n-2k-1)\Lambda G^{2k+2}(X_L G)^{n-2k-2},
\end{eqnarray*}
and using the identity
$$
\tfrac{n+2k+1}{n} \textstyle{\binom{n}{2k+1}}+ \tfrac{n-2k+1}{n}
 \textstyle{\binom{n}{2k-1}}=\textstyle{\binom{n+1}{2k+1}},
$$
by applying the recursion (\ref{rec}) we have that for $n=2i$
\begin{eqnarray*}
G_{2i+1}
&=& X_L(G)\,G_{2i}+\frac{1}{2i}G\,X_L(G_{2i})\\
&=& \sum_{k=0}^{i-1}\binom{2i}{2k+1}
\left(1+\frac{2k+1}{2i}\right) \Lambda^k G^{2k+1}(X_L G)^{2i-2k}
\\ & & + \sum_{k=0}^{i-1}\frac{2i-2k-1}{2i}\binom{2i}{2k+1}
 \Lambda^{k+1} G^{2k+3}(X_L G)^{2i-2k-2}\\
&=& \sum_{h=0}^{i-1}\binom{2i}{2h+1}
\frac{2i+2h+1}{2i} \Lambda^h G^{2h+1}(X_L G)^{2i-2h}
\\ & & + \sum_{h=1}^{i}\frac{2i-2h+1}{n}\binom{2i}{2h-1}
 \Lambda^h G^{2h+1}(X_L G)^{2i-2h}
 \\
 &=& (2i+1)G (X_L G)^{2i}+ \Lambda^{i}G^{2i+1} 
 \\ & &+ \sum_{h=1}^{i-1}\left[ \tfrac{2i+2h+1}{2i} \textstyle{\binom{2i}{2h+1}}+ \tfrac{2i-2h+1}{2i}
 \textstyle{\binom{2i}{2h-1}}\right]\Lambda^h G^{2h+1}(X_L G)^{2i-2h}
 \\
 &=& \sum_{k=0}^{i}
\binom{2i+1}{2k+1}\Lambda^k G^{2k+1}(X_L G)^{2i-2k},
\end{eqnarray*}
that is (\ref{Gn_alt}) with $n=2i+1$. The case $n=2i+1$ is analogous:
\begin{eqnarray*}
G_{2i+2}
&=& X_L(G)\,G_{2i+1}+\frac{1}{2i+1}G\,X_L(G_{2i+1})\\
&=& \sum_{k=0}^{i}\binom{2i+1}{2k+1}
\left(1+\frac{2k+1}{2i+1}\right) \Lambda^k G^{2k+1}(X_L G)^{2i+1-2k}
\\ & & + \sum_{k=0}^{i}\frac{2i-2k}{2i+1}\binom{2i+1}{2k+1}
 \Lambda^{k+1} G^{2k+3}(X_L G)^{2i-2k-1}\\
&=& \sum_{h=0}^{i}\binom{2i+1}{2h+1}
\frac{2i+2h+2}{2i+1} \Lambda^h G^{2h+1}(X_L G)^{2i+1-2h}
\\ & & + \sum_{h=1}^{i}\frac{2i-2h+1}{n}\binom{2i}{2h-1}
 \Lambda^h G^{2h+1}(X_L G)^{2i-2h}\\
\hphantom{G_{2i+2}} &=& (2i+2)G (X_L G)^{2i+1} \\ & &+ \sum_{h=1}^{i}\left[ \tfrac{2i+2h+2}{2i+1} \textstyle{\binom{2i+1}{2h+1}}+ \tfrac{2i-2h+2}{2i+1}\textstyle{\binom{2i+1}{2h-1}}\right]\Lambda^h G^{2h+1}(X_L G)^{2i+1-2h}
 \\
 &=& \sum_{k=0}^{i}
\binom{2i+2}{2k+1}\Lambda^k G^{2k+1}(X_L G)^{2i+1-2k},
\end{eqnarray*}
that is (\ref{Gn_alt}) with $n=2i+2$.
\end{proof}
It follows that, if $L$ is quadratic in the momenta, then  $G_n$ is  polynomial in the momenta of degree at most 
$$
n\left(1+deg(G)\right)-1.
$$

From \cite{CDR}, we can  derive the explicit expression of the first integrals $K_{m,n}$ of any $(m,n)$-extension

\begin{prop}
Let $G_n$ be a recursion with $G_1=G$ satisfying (\ref{CG}), that is with $\Lambda$ given by (\ref{Lamb}). We have 

\begin{equation}\label{EEcal}
K_{m,n}=P_{m,n}G_n+D_{m,n}X_{L}(G_n),
\end{equation}
with
$$
P_{m,n}=\sum_{k=0}^{[m/2]}\binom{m}{2k}\, \left(\frac mn \tilde \gamma \right)^{2k}p_u^{m-2k}\Lambda^k,
$$
$$
D_{m,n}=\frac 1{n}\sum_{k=0}^{[(m-1)/2]}\binom{m}{2k+1}\, \left(\frac mn \tilde \gamma \right)^{2k+1}p_u^{m-2k-1}\Lambda^k, \quad m>1,
$$
where $[\cdot]$ denotes the integer part and $D_{1,n}=\frac m{n^2}\tilde \gamma$.
\end{prop}
 
\footnote{
We remark that in Theorem 3 of \cite{CDR} the upper limit of the sum in $D_m$ is misprinted as $[m/2]-1$ instead of $[(m-1)/2]$.
}

\begin{proof}
Since the first integral $K_{m,n}$ of a $(m,n)$ extension of $L$ is the first integral of a $m$ extension of $L/n^2$, by Lemma \ref{Uno}  it follows that Theorem 3 of \cite{CDR} can be easily adapted, once we set  $\gamma \rightarrow \tilde \gamma m/n^2$, $-2m(cL+L_0) \rightarrow n^2\Lambda$ and $G\rightarrow G_n$. The result is straightforward.
\end{proof}

It follows that the degree of the first integral $K_{m,n}$, for $L$ quadratic in the momenta,  is 
$$
m+deg(G_n)\leq m+n\left(1+deg(G)\right)-1.
$$

When $m$ and $n$ are not reciprocally prime, one can expect that the polynomial $K_{m,n}$ factorizes into a number of factors, some of them again in the form $K_{r,s}$, where $m=a_0r$, $n=a_0s$, $a_0,r,s \in \mathbb N-\{0\}$. Indeed, the computation of several examples suggests that  $K_{r,s}$ is then  a divisor of $K_{m,n}$. 

\section{Examples of $(m,n)$-extensions}
Once a $m$-extension of a  Hamiltonian $L$ is known, it is straightforward to build any $(m,n)$-extension of $L$, as the following examples show.
\subsection{Two uncoupled oscillators}
Let us consider the two uncoupled harmonic oscillators described by the Hamiltonian (\ref{ex2}), extension of 
$$
L=\frac{1}{2}p_x^2+\omega^2x^2,
$$
for integer values of the parameter $\lambda$.
In this case the vector field $X_L$ is given by
$$
X_L=p_x\frac{\partial}{\partial x}-2\omega^2x\frac{\partial}{\partial p_x}.
$$
We can easily find a $G(x)$ not depending on the momenta  satisfying the condition (\ref{CG}),
with $\tilde c=0$, $\tilde L_0=\omega^2$; indeed (\ref{CG}) reduces to
$$
p_x^2(G'')-2\omega^2(xG'-G)=0.
$$
Thus, $G=x$ up to an inessential multiplicative constant and $X_L(G)=p_x$.
Then,  by setting $A=1$, we can construct the $(m,n)$-extension of $L$
$$
H_{m,n}=\frac 12\left( p_u^2+\left(\frac mn\right)^2p_x^2\right)+\omega^2\left(\frac mn \right)^2\left(x^2+u^2\right),
$$
for any $m,n$ positive integers, which is equivalent (up a rescaling and a constant factor) to (\ref{osc0})
with $\lambda=m/n$.
The first terms in the recursion (\ref{rec}) are
\begin{eqnarray*}
G_1&=&x,\\
G_2&=&2x\;p_x,\\
G_3&=&3x\;p_x^2-2\omega^2x^3,\\
G_4&=&4x\;p_x^3-8\omega^2x^3\;p_x,\\
G_5&=&5x\;p_x^4-20\omega^2x^3\;p_x^2+4\omega^4x^5.
\end{eqnarray*}
Alternatively, formula (\ref{Gn_alt}) becomes in this case
$$
G_n=\sum_{k=0}^{\left[\frac{n-1}{2}\right]}
\binom{n}{2k+1}(-2\omega^2)^k x^{2k+1}p_x^{n-2k-1}.
$$
For some values of $(m,n)$ we get the following first integrals
\begin{eqnarray*}
K_{1,1}&=& x\,p_u-u\,p_x,\\
K_{1,2}&=& 2x\;p_xp_u-u\left(\frac{1}{2}p_x^2-\omega^2 x^2\right),\\
K_{2,2}&=& 2(x\,p_u-u\,p_x)(p_x p_u+2\omega^2 xu),\\
K_{3,2}&=& 2x\;p_xp_u^3-\frac 92 u\;p_x^2p_u^2+9\omega^2 x^2u \;p_u^2 -27\omega^2xu^2\,p_xp_u,
\\
& & +\frac{27}4 \omega^2u^3\;p_x^2 -\frac{27}{2}\omega^4x^2u^3.
\end{eqnarray*}
We remark that $K_{2,2}$ is a multiple of $K_{1,1}$. This factorization is related to the fact that $(1,1)$ and $(2,2)$ represent the same rational number.

\subsection{Calogero-type systems}
We consider now the $(m,n)$-extensions of the Hamiltonian (\ref{calL}).
In this case the vector field $X_L$ is given by
$$
X_L=p_\phi\frac{\partial}{\partial \phi}+ \frac{2a\cos \phi}{\sin^3\phi} \frac{\partial}{\partial p_\phi}.
$$
We can easily find a $G(\phi)$ not depending on the momenta  satisfying the condition (\ref{CG}),
with $\tilde c=1$, $\tilde L_0=0$; indeed (\ref{CG}) splits into
$$
(G''+ G)=0, \qquad G'+ \frac{\sin \phi}{\cos\phi} G=0.
$$
Thus, up to an inessential multiplicative constant, a solution is $G=\cos\phi$  and $X_L(G)=p_\phi \sin \phi$.
Then, by (\ref{rec}) we can construct the sequence of $G_n$, whose first terms are
\begin{eqnarray*}
G_1&=&\cos\phi,\\
G_2&=&-\sin 2\phi\; p_\phi,\\
G_3&=&-\cos 3\phi\; p_\phi^2-\frac{2a\cos^3 \phi}{\sin^2\phi},\\
G_4&=&\sin 4\phi\; p_\phi^3-\frac{8a\cos^3 \phi}{\sin\phi}p_\phi,\\
G_5&=&\cos 5\phi\; p_\phi^4+\frac{4a(6\cos^2\phi-5)\cos^3\phi}{\sin^2\phi}\,p_\phi^2+\frac{4a^2\cos^5 \phi}{\sin^4\phi}.
\end{eqnarray*}
We can construct the $(m,n)$-extension of $L$
$$
H_{m,n}=\frac 12 p_u^2 +\frac{m^2}{n^2S^2_\kappa (u)} \left(\frac 12p_\phi^2+\frac {a}{\sin^2(\phi)}\right),
$$
for any $m,n$ positive integers. According to the sign of the parameter $\kappa$, the extended Hamiltonian is defined on the cotangent bundle of a sphere ($\kappa>0$), of a plane ($\kappa=0$) or of a pseudo-sphere ($\kappa<0$).
Moreover, for $\kappa=0$ the extension 
coincides with (\ref{ex1}) that, after the rescaling $\phi= \lambda \psi$, becomes
the Hamiltonian (\ref{cal0})
$$
H=\frac{1}{2}p_u^2+\frac{1}{u^2}\left(\frac{1}{2}p_\psi^2+\frac{a}{\sin^2(
\lambda\psi)}\right),
$$ 
for any rational $\lambda=m/n$. This Hamiltonian generalizes the Calogero three particle chain without harmonic term (obtained for $\lambda=3$, see \cite{CDR0}, also known   as Jacobi system) and it was the starting point of our work in \cite{Pol}. This is a particular case of the TTW system \cite{TTW}.
 
As an example, we compute a first integral of the extension (\ref{ex1})  for $(m,n)=(2,3)$ i.e., $\lambda=2/3$, 
\begin{eqnarray*}
K_{2,3}\! \!  \!&=& \! \!\!  -\cos3\phi\; p_u^2p_\phi^2 +\frac {4}{3u}\sin\phi(4\cos^2\phi-1)\; p_u p_\phi^3+\frac{4\cos3\phi}{9u^2}p_\phi^4 
 -\frac{2a\cos^3\phi}{\sin^2\phi} p_u^2 
 \\
 & &
  + \frac{8a\cos^2\phi}{u\sin\phi}\, p_up_\phi  
 +
\frac{8a(5\cos^2\phi-3)\cos\phi} {9u^2\sin^2\phi}\, p_\phi^2 +\frac{16a^2\cos^3\phi}{9u^2\sin^4\phi}.
\end{eqnarray*}
%
%
%
%
%
%
%
%

\subsection{Three-sphere}
Any example of $m$-extension taken from \cite{Pol,CDRfi,CDR,CDR1} can be easily transformed into a $(m,n)$-extension. For example, let us consider from \cite{CDRfi} the $m$-extension of the geodesic Hamiltonian 
$$
L=\frac{1}{2}\left(p^2_\eta+\frac{p^2_{\xi_1}}{\sin^2\eta}+\frac{p^2_{\xi_2}}{\cos^2\eta}\right)
$$ 
on the three-sphere $\mathbb S^3$,  with coordinates $(\eta, \xi_1, \xi_2)$ where $0<\eta<\pi/2$, $0\leq \xi_i <2\pi$ 
and the parametrization in $\mathbb R^4$  given by $$x=\cos\xi_1 \sin\eta,\ y=\sin\xi_1 \sin\eta,\ z=\cos\xi_2 \cos\eta,\ t=\sin \xi_2 \cos\eta.$$ 
Equation (\ref{CG}) admits a complete solution (i.e., a solution depending on the maximal number of parameters $a_i$) if and only if $mc=\tilde{c}$ equals the curvature $K=1$ of the three-sphere \cite{CDRfi}.The complete solution is
\begin{eqnarray*}
G&=&(a_3\sin(\xi_1)+a_4\cos(\xi_1))\sin(\eta)+(a_1\sin(\xi_2)+a_2\cos(\xi_2))\cos(\eta)\\
&=& a_4x+a_3y+a_2z+a_1t,
\end{eqnarray*}
with $a_i$ constants. 
We look for compatible potentials $V$, i.e., functions that can be added to $L$ in such a way  that  $G$ satisfies (\ref{CG})  for the natural Hamiltonian $L+V$.
It is easier to find compatible $V$  when some of the parameters $a_i$  in $G$ are assigned.
For example, if 
$a_2=a_3=a_4=0$, a compatible potential is
$$
V=\frac{\sin \xi_1}{\cos \xi_2\cos \eta \sin \eta}.
$$
With this choice of $(a_i)$, the first three $G_n$ obtained from $G=\sin \xi_2\cos \eta$ are 
\begin{eqnarray*}
G_1&=&\sin \xi_2 \cos \eta,\\
G_2&=&-\sin 2\eta \sin ^2\xi_2\;p_\eta+2\sin \xi_2\cos\xi_2\;p_{\xi_2},\\
G_3&=&-\cos 3\eta \sin ^3 \xi_2\;p_\eta^2-6 \sin \eta\sin ^2 \xi_2 \cos \xi_2\; p_\eta p_{\xi_2} -\frac{\cos^3 \eta}{\sin^2 \eta}\sin^3 \xi_2\; p_{\xi_1}^2
\\
&& +
\frac{\sin \xi_2(3\cos^2 \xi_2-\cos^2 \eta
\sin^2 \xi_2)}{\cos \eta}\,p_{\xi_2}^2
%
 -2\frac{\sin^3\xi_2\sin\xi_1\cos^2\eta}{\cos\xi_2\sin\eta},
\end{eqnarray*}
the $(m,n)$-extension of $L+V$ is 
$$
H_{m,n}=\frac 12 p_u^2+\left(\frac {m}{n}\right)^2\frac{L+V} {S^2_\kappa (u)},
$$
and, for example,  a first integral of $H_{1,2}$ is
\begin{eqnarray*}
K_{1,2}&=&\frac 1{T_\kappa( u)}\left(-\frac 12 \cos 2\eta\sin^2 \xi_2\;p_\eta^2 -\frac{\sin \eta \cos \xi_2 \sin \xi_2}{\cos \eta}\,p_\eta p_{\xi_2}\right.\\
&-&\left. \frac{\cos^2 \eta \sin^2\xi_2}{2\sin^2\eta}\,p_{\xi_1}^2 +\frac{\cos^2 \xi_2-\cos^2\eta \sin^2 \xi_2}{2\cos^2 \eta}\,p_{\xi_2}^2- \frac{\sin^2\xi_2\sin \xi_1 \cos \eta}{\sin \eta \cos \xi_2}\right)\\
&-&2\sin 2 \eta \sin^2\xi_2 \; p_u p_\eta+2\sin\xi_2\cos \xi_2\;p_up_{\xi_2}.
\end{eqnarray*}

\section{Conclusions}

In this article we show that the technique of extensions can be modified in order to be applied successfully also to a class of Hamiltonian systems (including some known superintegrable ones) 
depending on rational values of a parameter.  
Many properties of the first integrals obtained, for instance their factorization,  need a deeper analysis and the possibility of the use of the technique to build new extended systems with rational parameters is still unexplored. This short exposition is certainly not a complete theory of extensions with rational parameters, but it represents a solid ground for building such a theory.
 
\section{Appendix}

The {\it trigonometric tagged functions} 
$$
S_\kappa(x)=\left\{\begin{array}{ll}
\frac{\sin\sqrt{\kappa}x}{\sqrt{\kappa}} & \kappa>0 \\
x & \kappa=0 \\
\frac{\sinh\sqrt{|\kappa|}x}{\sqrt{|\kappa|}} & \kappa<0
\end{array}\right.
\qquad
C_\kappa(x)=\left\{\begin{array}{ll}
\cos\sqrt{\kappa}x & \kappa>0 \\
1 & \kappa=0 \\
\cosh\sqrt{|\kappa|}x & \kappa<0
\end{array}\right.,
$$
$$
T_\kappa(x)=\frac {S_\kappa(x)}{C_\kappa(x)},
$$
are employed, explicitly or not, by several scholars  (see  \cite{Tag}, \cite{Chern}) and appear in several branches of mathematics.

The trigonometric tagged functions satisfy a number of properties  analogous to those of  ordinary trigonometric functions.  Their main advantage is to unify trigonometric and hyperbolic functions in a homogeneous way. From the definition it follows that $S_\kappa(x)$ is an odd function while $C_\kappa(x)$ is an even function but, if $\kappa\le0$ and $x\in\mathbb{R}$ they are no more periodic. The basic properties are
$$
C_\kappa^2(x)+\kappa S_\kappa^2(x)=1,
$$
\begin{eqnarray*}
S_\kappa(x\pm y)&=&S_\kappa(x)C_\kappa(y)\pm C_\kappa(x)S_\kappa(y),\\
C_\kappa(x\pm y)&=&C_\kappa(x)C_\kappa(y)\mp\kappa S_\kappa(x)S_\kappa(y),
\end{eqnarray*}
from these the duplication formulas can be obtained
\begin{eqnarray*}
S_\kappa(2x)&=&2S_\kappa(x)C_\kappa(x),\\
C_\kappa(2x)&=&C_\kappa^2(x)-\kappa S_\kappa^2(x)=
\left\{\begin{array}{l}
2C_\kappa^2(x)-1,\\
1-2\kappa S_\kappa^2(x).
\end{array}\right.
\end{eqnarray*}
The bisection formula
$$
C_\kappa^2(x)=\frac{1+C_\kappa(2x)}{2}
$$
is always true, while one has to set $\kappa\neq0$  in order to obtain
$$
S_\kappa^2(x)=\frac{1-C_\kappa(2x)}{2\kappa}.
$$
The functions $S_\kappa(x)$ and $C_\kappa(x)$ are related through differentiation:
$$
\frac{d}{dx}S_\kappa(x)=C_\kappa(x) \qquad \frac{d}{dx}C_\kappa(x)=-\kappa S_\kappa(x).
$$
Hence, the linear combinations of $S_\kappa(x)$ and $C_\kappa(x)$ provide the general solution of the differential equation $${F}''+\kappa F=0.$$

\subsection*{Acknowledgement}
 This work has been partially (CC) supported by PRIN 2010/2011 Research project ``Teorie geometriche e analitiche dei sistemi Hamiltoniani in dimensioni finite e infinite" 
grant no 2010JJ4KPA\_007.

\end{document}